\documentclass[12pt]{amsart}
\usepackage{amssymb,amsmath,mathrsfs,enumerate,mathtools}
\usepackage{tikz-cd}
\usepackage[pdftex,bookmarks=true]{hyperref}

\theoremstyle{plain}

\newtheorem{theorem}{Theorem}[section]

\newtheorem{proposition}[theorem]{Proposition}
\newtheorem{lemma}[theorem]{Lemma}

\theoremstyle{definition}

\numberwithin{equation}{section}

\DeclareRobustCommand{\rchi}{{\mathpalette\irchi\relax}}
\newcommand{\irchi}[2]{\raisebox{\depth}{$#1\chi$}}

\author[M.~Tantrawan]{Made Tantrawan}
\address{\textit{a} Department of Mathematics,   National University of Singapore, Singapore 119076 \newline
	\textit{b} Department of Mathematics, Faculty of Mathematics and Natural Sciences, Universitas Gadjah Mada, Indonesia 55281}
\email{made.tantrawan@ugm.ac.id}

\author[D.~Leung]{Denny H.~Leung}
\address{Department of Mathematics, National University of Singapore, Singapore 119076}
\email{matlhh@nus.edu.sg}

\title[On closedness of law-invariant convex sets]{On closedness of law-invariant convex sets in rearrangement invariant spaces}

\thanks{The first author is supported by NUS Research Scholarship. The second author is partially supported by AcRF grant R-146-000-242-114.}
\keywords{convex sets, law-invariant, rearrangement invariant spaces, Fatou property}

\subjclass[2010]{46A55, 46E30, 46A20}

\date{\today}

\begin{document}

\begin{abstract}
This paper presents relations between several types of closedness of a law-invariant convex set in a rearrangement invariant space $\mathcal{X}$. In particular, we show that order closedness, $\sigma(\mathcal{X},\mathcal{X}_n^\sim)$-closedness and $\sigma(\mathcal{X},L^\infty)$-closedness of a law-invariant convex set in $\mathcal{X}$ are equivalent, where $\mathcal{X}_n^\sim$ is the order continuous dual of $\mathcal{X}$. We also provide some application to proper quasiconvex law-invariant functionals with the Fatou property.
\end{abstract}

\maketitle

\section{Introduction}

\subsection{Background}
An important problem arising from the theory of risk measures asks whether order closedness of a convex set in a Banach function space $\mathcal{X}$ ensures closedness with respect to the $\sigma(\mathcal{X},\mathcal{X}_n^\sim)$-topology (\cite{Ow}). A positive answer to this problem will guarantee the admittance of Fenchel-Moreau dual representation of a proper convex functional (in particular, a coherent risk measure) with the Fatou property. It is known that the problem has a positive answer for $L^p$ spaces (see \cite{De} for $\mathcal{X}=L^\infty$). However, it may have a negative answer in general Orlicz spaces (\cite{GLX}). In 2018, Gao et al. \cite{GLMX} showed that the problem still has a positive answer for any Orlicz space when the convex set is assumed to be law-invariant. One of the main results of the present paper extends this to general rearrangement invariant (r.i.) spaces, the largest class of Banach function spaces where the law-invariant condition is applicable.

We organize our paper as follows. In Section 2, we show that for an r.i. space $\mathcal{X}$, order closedness, $\sigma(\mathcal{X},\mathcal{X}_n^\sim)$-closedness and $\sigma(\mathcal{X},L^\infty)$-closedness of law-invariant convex sets in $\mathcal{X}$ are equivalent. We also characterize when order closedness and norm closedness of a law-invariant convex set are equivalent. In Section 3, we apply the above results to obtain relations between the Fatou property and other types of lower semicontinuity, namely $\sigma(\mathcal{X},\mathcal{X}_n^\sim)$-lower semicontinuity, $\sigma(\mathcal{X},\mathcal{X}_{uo}^\sim)$-lower semicontinuity $\sigma(\mathcal{X},L^\infty)$-lower semicontinuity, norm lower semicontinuity, and the strong Fatou property.

\subsection{Banach function spaces and rearrangement invariant spaces}

Throughout this paper, we always work on a nonatomic probability space $(\Omega,\Sigma,\mu)$. Let $L^0=L^0(\Omega,\Sigma,\mu)$ be the vector lattice of all (equivalence classes with respect to equality a.e. of) real measurable functions on $\Omega$. A Banach function space $\mathcal{X}$ over $\Omega$ is an ideal of $L^0$ endowed with a complete norm such that $\|f\|_{\mathcal{X}}\leq \|g\|_{\mathcal{X}}$ whenever $f,g\in \mathcal{X}$ and $|f|\leq |g|$. 

A sequence $\{f_n\}$ in $\mathcal{X}$ is said to order converge to $f\in\mathcal{X}$, written as $f_n\xrightarrow{o}f$, if $f_n\xrightarrow{a.e.}f$ and $|f_n|\leq h$ for some $h\in\mathcal{X}$. The order continuous dual $\mathcal{X}_n^\sim$ (resp., unbounded order continuous dual $\mathcal{X}_{uo}^\sim$) of $\mathcal{X}$ is the collection of all linear functionals $\varphi$ on $\mathcal{X}$ such that $\varphi(f_n)\to0$ whenever $f_n\xrightarrow{o}0$ (resp., $f_n\xrightarrow{a.e.}0$ and $\{f_n\}$ is norm bounded). Both $\mathcal{X}_n^\sim$ and $\mathcal{X}_{uo}^\sim$ are ideals in the norm dual $\mathcal{X}^*$ of $\mathcal{X}$. Denote by $\mathcal{X}_a$ the order continuous part of $\mathcal{X}$, i.e., the set of all $f\in\mathcal{X}$ such that $\|f\rchi_{A_n}\|_{\mathcal{X}}\to0$ whenever $A_n\downarrow\emptyset$. For any Banach function space $\mathcal{X}$, $\mathcal{X}_{uo}^\sim$ is the order continuous part of $\mathcal{X}_n^\sim$ (\cite[Theorem 2.3.]{GLX1}). A Banach function space $\mathcal{X}$ is order continuous if $\mathcal{X}_a=\mathcal{X}$, or equivalently, $\mathcal{X}_n^\sim=\mathcal{X}^*$ (\cite[Theorem 2.4.2]{MN}). The order continuous dual $\mathcal{X}_n^\sim$ is isometrically isomorphic to the associated space $\mathcal{X}'$ of $\mathcal{X}$, i.e., the space of all $g\in L^0$ such that the associate norm 
\[
\|g\|_{\mathcal{X}'}=\sup\left\{\int_{\Omega}fg\mathrm{d}\mu:\|f\|_{\mathcal{X}}\leq 1\right\}
\] 
is finite (\cite[Theorem 2.6.4]{MN}). We will identifty $\mathcal{X}_n^\sim$ with the Banach function space $\mathcal{X}'$.

A rearrangement invariant (r.i.) space $\mathcal{X}$ is a Banach function space over $\Omega$ such that $\mathcal{X}\neq\{0\}$ and for every $g\in L^0$, $g\in \mathcal{X}$ and $\|g\|_{\mathcal{X}}=\|f\|_{\mathcal{X}}$ whenever $g$ has the same distribution as some $f\in \mathcal{X}$. Orlicz spaces, Lorentz spaces, Marcinkiewicz spaces, and Orlicz-Lorentz spaces are some examples of r.i. spaces. For any r.i. space $\mathcal{X}$, $L^\infty\subseteq \mathcal{X}\subseteq L^1$ (\cite[Corollary 6.7, p. 78]{BS}). Since $\mathcal{X}_n^\sim$ is an r.i. space (\cite[Proposition 4.2, p. 59]{BS}), we also have $L^\infty \subseteq \mathcal{X}_n^\sim \subseteq L^1$. A set $E$ in an r.i. space $\mathcal{X}$ is said to be law-invariant if $g\in E$ for every $g\in \mathcal{X}$ that has the same distribution as some $f\in E$. 

We refer to \cite{BS, KPS, MN}\footnote{In \cite{BS}, all Banach function spaces $\mathcal{X}$ are assumed to satisfy the following condition: $f\in \mathcal{X}$ and $\|f\|_{\mathcal{X}}=\sup_n\|f_n\|_{\mathcal{X}}$ whenever $\{f_n\}$ is an increasing norm bounded sequence in $\mathcal{X}_+$ and $f$ is the pointwise limit of $\{f_n\}$. The results we cite from \cite{BS} remain true without assuming this extra condition.} for more details on Banach function spaces and rearrangement invariant spaces.

\section{Closedness of law-invariant convex sets}

\textit{For the rest of the paper, we always assume that $\mathcal{X}$ is an r.i. space}. Recall that $E \subseteq \mathcal{X}$ is order closed if its order closure $$\overline{E}^o=\{f\in\mathcal{X}:\exists \{f_n\} \subseteq E\ \text{s.t.}\ f_n\xrightarrow{o}f\}$$ is equal to $E$ itself. Observe that for any $E \subseteq \mathcal{X}$, $$\overline{E}^{\|\cdot\|_{\mathcal{X}}}\subseteq \overline{E}^o \subseteq \overline{E}^{\sigma(\mathcal{X},\mathcal{X}_n^\sim)} \subseteq \overline{E}^{\sigma(\mathcal{X},L^\infty)}.$$ Hence, every $\sigma(\mathcal{X},L^\infty)$-closed set is $\sigma(\mathcal{X},\mathcal{X}_n^\sim)$-closed, every  $\sigma(\mathcal{X},\mathcal{X}_n^\sim)$-closed set is order closed, and every order closed set is norm closed.

We will write $\pi$ to denote a finite measurable partition of $\Omega$ whose members have non-zero measures. Denote by $\sigma(\pi)$ the finite $\sigma$-subalgebra generated by $\pi$. For any $f\in\mathcal{X}$,
\[
\mathbb{E}[f|\pi]:=\mathbb{E}[f|\sigma(\pi)]=\sum_{i=1}^n\frac{\int_{\Omega}f\rchi_{A_i}\mathrm{d}\mu}{\mu(A_i)}\rchi_{A_i}\in L^\infty,
\]
where $\pi=\{A_1,A_2,\dotsc,A_n\}$. The collection $\Pi$ of all such $\pi$'s is directed by refinement. When $\mathcal{X}=L^\infty$, it is known that $\mathbb{E}[f|\pi]\xrightarrow{\|\cdot\|_{\infty}}f$ (see \cite{JST} and \cite{Sv}), in particular, any $f\in L^\infty$ is an order limit of a sequence of $\{\mathbb{E}[f|\pi]\}$. The last property is also true for Orlicz spaces, that is, any element $f$ in an Orlicz space $\mathcal{X}$ is an order limit of a sequence of $\{\mathbb{E}[f|\pi]\}$ (\cite{GLMX}). Until now, it is unknown whether the same property holds in general r.i. spaces (see \cite{CGX}). However, the following weaker property does hold for any r.i. space.

\begin{proposition}\label{ri-orderdualcnvseq}
	For every $f\in \mathcal{X}$, there exist sequences $\{\pi_n\} \subseteq \Pi$ and $\{f_n\} \subseteq\mathcal{X}$ such that $0\leq |f_n|\leq |f|$ for every $n$, $\mathbb{E}[f-f_n|\pi_n]\xrightarrow{o}f$ and $\mathbb{E}[f_n|\pi_n]\xrightarrow{\sigma(\mathcal{X},\mathcal{X}_n^\sim)}0$.
\end{proposition}
\begin{proof}
	Let $f\in \mathcal{X}$. Put $f_n=f\rchi_{\{|f|>n\}}$ for every $n$. Observe that $f-f_n\in L^\infty$ for every $n$. Then we can find $\pi_n\in\Pi$ such that
	\[
	\left\|\mathbb{E}[f-f_n|\pi_n]-(f-f_n)\right\|_\infty\leq\frac{1}{n}
	\]
	for every $n$. It follows that $\mathbb{E}[f-f_n|\pi_n]-(f-f_n)\xrightarrow{o}0$. Since $f_n\xrightarrow{o}0$, we obtain that $\mathbb{E}[f-f_n|\pi_n]\xrightarrow{o}f$.
	
	For any $h\in L^0$, denote by $h^*:[0,\infty)\to\mathbb{R}$ the decreasing rearrangement of $h$, that is,
	\[
	h^*(t)=\inf\{\lambda:\mu(\{x\in\Omega:|h(x)|>\lambda\})\leq t\},\ t\geq 0.
	\]
	By \cite[Proposition 3.7, p. 57]{BS}, 
	\[
	\int_{0}^a\left(\mathbb{E}[f_n|\pi_n]\right)^*(t)\mathrm{d}t\leq \int_{0}^a\left(f_n\right)^*(t)\mathrm{d}t
	\]
	for any $a\in[0,1]$. Let $g\in\mathcal{X}_n^\sim$. From Hardy's Lemma (\cite[Proposition 3.6, p. 56]{BS}), we have that
	\[
	\int_{0}^{1}\left(\mathbb{E}[f_n|\pi_n]\right)^*(t)g^*(t)\mathrm{d}t\leq \int_{0}^{1}\left(f_n\right)^*(t)g^*(t)\mathrm{d}t.
	\]
	Since $f_n\xrightarrow{a.e.}0$ and $|f_n|\leq |f|$, we see that $f_n^*\xrightarrow{a.e.}0$ and $0\leq f_n^*\leq f^*$ (see properties $2^\circ$ and $12^\circ$ in \cite[pp. 63-67]{KPS}). Since $\mathcal{X}$ is an r.i. space, $\int_0^{1}f^*(t)g^*(t)\mathrm{d}t<\infty$. By Hardy-Littlewood inequality (\cite[Proposition 2.2, p. 44]{BS}) and  the dominated convergence theorem, we obtain that
	\begin{eqnarray*}
		\left|\int_{\Omega}\mathbb{E}[f_n|\pi_n]g\mathrm{d}\mu\right|&\leq&\int_{0}^{1}\left(\mathbb{E}[f_n|\pi_n]\right)^*(t)g^*(t)\mathrm{d}t\\
		&\leq&\int_{0}^{1}\left(f_n\right)^*(t)g^*(t)\mathrm{d}t\to0.
	\end{eqnarray*}
	Thus, we conclude that $\mathbb{E}[f_n|\pi_n]\xrightarrow{\sigma(\mathcal{X},\mathcal{X}_n^\sim)}0$.
\end{proof}

Observe that the following properties hold for any r.i. space $\mathcal{X}$.

\begin{lemma}\label{Linfty-o.c.part} ${}$
	\begin{enumerate}[$(1)$]
		\item Either $\mathcal{X}=L^\infty$ or $L^\infty \subseteq\mathcal{X}_a$.
		\item Either $\mathcal{X}=L^1$ or $L^\infty \subseteq \mathcal{X}^\sim_{uo}$.
	\end{enumerate}
\end{lemma}
\begin{proof}
	The first part is a direct consequence of \cite[Lemma A.2]{CGX}. In fact, if $L^\infty\not \subseteq\mathcal{X}_a$, then $\rchi_{\Omega}\notin \mathcal{X}_a$ or equivalently $\lim_{\mu(A)\to0}\|\rchi_A\|_{\mathcal{X}}>0$. By \cite[Lemma A.2]{CGX}, we conclude that $\mathcal{X}=L^\infty$. For the second part, suppose that $L^\infty\not \subseteq\mathcal{X}_{uo}^\sim$. We have that $\rchi_{\Omega}\notin \mathcal{X}_{uo}^\sim= (\mathcal{X}_n^\sim)_a$ and hence, $\lim_{\mu(A)\to0}\|\rchi_A\|_{\mathcal{X}_n^\sim}>0$. By \cite[Lemma A.2]{CGX}, $\mathcal{X}_n^\sim=L^\infty$ with equivalence of norms. Clearly, $\mathcal{X}\neq L^\infty$. By the first part, $L^\infty \subseteq \mathcal{X}_a$ and hence, $(\mathcal{X}_a)_n^\sim=(\mathcal{X}_a)^*=\mathcal{X}_n^\sim=L^\infty$ by (\cite[Corollary 4.2, p.~23]{BS}). Since $(\mathcal{X}_a)^*$ is isomorphic to the abstract $M$-space $L^\infty$, $\mathcal{X}_a$ is isomorphic to an abstract $L$-space by \cite[Proposition 1.4.7]{MN}. It follows from \cite[Proposition 8.2, p.~113]{Sc} that $\mathcal{X}_a$ is monotonically complete, i.e., $\sup_\alpha f_\alpha\in \mathcal{X}$ for every increasing norm bounded net $\{f_\alpha\}$ in $\mathcal{X}_+$. Hence,  $\mathcal{X}_a=((\mathcal{X}_a)_n^\sim)_n^\sim=L^1$ by \cite[Theorem 2.4.22]{MN}. Since $L^1=\mathcal{X}_a \subseteq \mathcal{X}  \subseteq L^1$, we conclude that $\mathcal{X}=L^1$.
\end{proof}

Using Proposition \ref{ri-orderdualcnvseq}, Lemma \ref{Linfty-o.c.part} and \cite[Proposition 2]{CGX}, we obtain the following generalization of \cite[Corollary 4.5]{GLMX} from the class of Orlicz spaces to all r.i. spaces. 

\begin{theorem}\label{invariant-0closed-cvx-prop}
	Let $C$ be a convex order closed law-invariant set in $\mathcal{X}$. Then $f\in C$ if and only if $\mathbb{E}[f|\pi]\in C$ for any $\pi\in \Pi$. 
\end{theorem}
\begin{proof}
	We only need to show that  $\mathbb{E}[f|\pi]\in C$ for any $\pi\in \Pi$ implies $f\in C$ as the reverse implication is already given in \cite[Proposition 2]{CGX}. Suppose that $f\in\mathcal{X}$ and $\mathbb{E}[f|\pi]\in C$ for any $\pi\in \Pi$. If $\mathcal{X}=L^\infty$, $f$ is an order limit of a sequence of $\{\mathbb{E}[f|\pi]\}$. Since $C$ is order closed, we deduce that $f\in C$. Now suppose that $\mathcal{X}\neq L^\infty$. By Lemma \ref{Linfty-o.c.part}, $L^\infty \subseteq\mathcal{X}_a$. According to Proposition \ref{ri-orderdualcnvseq}, there exist sequences $\{\pi_n\} \subseteq \Pi$ and $\{f_n\} \subseteq \mathcal{X}$ such that $0\leq |f_n|\leq |f|$ for every $n$, $\mathbb{E}[f-f_n|\pi_n]\xrightarrow{o}f$ and $\mathbb{E}[f_n|\pi_n]\xrightarrow{\sigma(\mathcal{X},\mathcal{X}_n^\sim)}0$. Observe that $\mathbb{E}[f_n|\pi_n]\in L^\infty \subseteq\mathcal{X}_a$ for every $n$. Since $\mathcal{X}_n^\sim$ can be canonically identified with $\mathcal{X}_a^*$ (\cite[Corollary 4.2, p. 23]{BS}), $\{\mathbb{E}[f_n|\pi_n]\}$ converges weakly to $0$ in $\mathcal{X}_a$ and hence in $\mathcal{X}$. 
	
	Denote by $\text{co}(E)$ the set of all convex combinations of elements in $E$. Since $\{\mathbb{E}[f_n|\pi_n]\}$ converges weakly to $0$, we obtain that 
	\[
	0\in\overline{\text{co}(\{\mathbb{E}[f_n|\pi_n]:n\geq k\})}^{\sigma(\mathcal{X},\mathcal{X}^*)}=\overline{\text{co}(\{\mathbb{E}[f_n|\pi_n]:n\geq k\})}^{\|\cdot\|_{\mathcal{X}}}
	\]
	for every $k\geq 1$. Then we can find a strictly increasing sequence $\{p_k\}$ and $h_k\in \text{co}(\{\mathbb{E}[f_n|\pi_n]:p_{k-1}< n\leq p_k\})$ $(p_0:=0)$ such that $\|h_k\|_{\mathcal{X}}\leq \frac{1}{2^k}$ for every $k$. It follows that $h_k \xrightarrow{o}0$. Write 
	\[
	h_k=\sum_{n=p_{k-1}+1}^{p_k}a_{kn}\mathbb{E}[f_n|\pi_n]
	\]
	where $0\leq a_{kn}\leq 1$ and $\displaystyle{\sum_{n=p_{k-1}+1}^{p_k}a_{kn}=1}$ for every $k$. Since $\mathbb{E}[f-f_n|\pi_n]\xrightarrow{o}f$,
	\[
	\sum_{n=p_{k-1}+1}^{p_k}a_{kn}\mathbb{E}[f-f_n|\pi_n]\xrightarrow{o}f.
	\]
	Therefore,
	\[
	\sum_{n=p_{k-1}+1}^{p_k}a_{kn}\mathbb{E}[f|\pi_n]=\left(\sum_{n=p_{k-1}+1}^{p_k}a_{kn}\mathbb{E}[f-f_n|\pi_n]\right) + h_k\xrightarrow{o}f.
	\]
	Since $\sum_{n=p_{k-1}+1}^{p_k}a_{kn}\mathbb{E}[f|\pi_n]\in C$ for every $k$ and $C$ is order closed, we conclude that $f\in C$.

\end{proof}

Now, we are already to prove the main result of this section. This result generalizes Corollary 4.6 in \cite{GLMX}. 

\begin{theorem}\label{cvxinvariantequiv}
	The following are equivalent for a convex law-invariant set $C$ in $\mathcal{X}$:
	\begin{enumerate}[$(1)$]
		\item $C$ is order closed.
		\item $C$ is $\sigma(\mathcal{X},\mathcal{X}_n^\sim)$-closed.
		\item $C$ is $\sigma(\mathcal{X},L^\infty)$-closed.
	\end{enumerate}
	If $\mathcal{X}\neq L^1$, they are also equivalent to the following:
	\begin{enumerate}[$(1)$]\setcounter{enumi}{3}
		\item $C$ is $\sigma(\mathcal{X},\mathcal{X}_{uo}^\sim)$-closed.
		
	\end{enumerate}
\end{theorem}
\begin{proof}
	As stated at the beginning of the section, $(3)\implies (2)\implies (1)$ always holds. If $\mathcal{X}\neq L^1$,  $L^\infty \subseteq \mathcal{X}_{uo}^\sim \subseteq \mathcal{X}_n^\sim$ by Lemma \ref{Linfty-o.c.part} and hence, $(3)\implies (4)\implies (2)$. Thus, it is enough to show that $(1)\implies(3)$. Suppose that $C$ is order closed. Let $f$ be an element in the $\sigma(\mathcal{X},L^\infty)$-closure of $C$. There is a net $\{f_\alpha\} \subseteq C$ that $\sigma(\mathcal{X},L^\infty)$-converges to $f$. Observe that $\{\mathbb{E}[f_\alpha|\pi]\}$ converges weakly to $\mathbb{E}[f|\pi]$. Indeed, for a partition $\pi=\{A_1,\dotsc,A_n\}\in\Pi$ and $\varphi\in\mathcal{X}^*$, we have
	\[
	\varphi(\mathbb{E}[f_\alpha|\pi])=\int_{\Omega}f_\alpha\sum_{i=1}^n\frac{\varphi(\rchi_{A_i})}{\mu(A_i)}\rchi_{A_i}\mathrm{d}\mu\to\int_{\Omega}f\sum_{i=1}^n\frac{\varphi(\rchi_{A_i})}{\mu(A_i)}\rchi_{A_i}\mathrm{d}\mu=\varphi(\mathbb{E}[f|\pi]).
	\]
	Since $C$ is order closed, it is weakly closed. Together with the fact that each $\mathbb{E}[f_\alpha|\pi]$ lies in $C$ (Theorem \ref{invariant-0closed-cvx-prop}), we obtain that $\mathbb{E}[f|\pi]\in C$ for all $\pi\in \Pi$ and hence, $f\in C$. This proves that $C$ is $\sigma(\mathcal{X},L^\infty)$-closed.
\end{proof}

Lemma \ref{Linfty-o.c.part} can also be used to characterize the equivalence between order closedness and norm closedness of law-invariant convex sets.

\begin{theorem}\label{cvxinveqvnorm}
	The following are equivalent:
	\begin{enumerate}[$(1)$]
		\item Every norm closed law-invariant convex set in $\mathcal{X}$ is order closed. 
		\item Either $\mathcal{X}=L^\infty$ or $X$ is order continuous.
	\end{enumerate}
\end{theorem}
\begin{proof}
	$(1)\implies(2)$. Suppose that $\mathcal{X}\neq L^\infty$ and $\mathcal{X}$ is not order continuous. We will show that there exists a norm closed law-invariant convex set in $X$ which is not order closed.  First, observe that $\mathcal{X}_a$ is law-invariant. To see this, let $g\in \mathcal{X}$ have the same distribution as some $f\in \mathcal{X}_a$. For every $n$, both $f\rchi_{\{|f|\leq n\}}$ and $g\rchi_{\{|g|\leq n\}}$ are in $L^\infty$ and they have the same distribution. It follows that $$\lim_{n\to\infty}\left\|g\rchi_{\{|g|\leq n\}}\right\|_{\mathcal{X}}=\lim_{n\to\infty}\left\|f\rchi_{\{|f|\leq n\}}\right\|_{\mathcal{X}}=\left\|f\right\|_{\mathcal{X}}=\left\|g\right\|_{\mathcal{X}}$$
	by order continuity of $\mathcal{X}_a$. Since $L^\infty\subseteq \mathcal{X}_a$ by Lemma \ref{Linfty-o.c.part} and $\mathcal{X}_a$ is norm closed, we deduce that $g\in \mathcal{X}_a$. Hence, $\mathcal{X}_a$ is law-invariant. 
	
	Now, set 
	\[
	C=\left\{f\in\mathcal{X}_a: \int_{\Omega}f\mathrm{d}\mu\geq0\right\}.
	\]
	Clearly, $af+bg\in C$ for every $f,g\in C$ and $a,b\geq0$. In particular, $C$ is convex. Since $\mathcal{X}_a$ is law-invariant, $C$ is also law-invariant. Now let $f$ be an element in the norm closure of $C$. Then there is a sequence $\{f_n\} \subseteq C$ which norm converges to $f$. Since $\rchi_{\Omega}\in\mathcal{X}^*$, 
	\[
	\left|\int_{\Omega} f_n\mathrm{d}\mu-\int_{\Omega} f\mathrm{d}\mu\right|\leq \left\|f_n-f\right\|_{\mathcal{X}}\left\|\rchi_{\Omega}\right\|_{\mathcal{X}^*}\to0.
	\]
	Hence, $\int_{\Omega}f\mathrm{d}\mu\geq0$. Since $\mathcal{X}_a$ is norm closed, we see that $f\in \mathcal{X}_a$. Therefore, $f\in C$. This shows that $C$ is norm closed.
	
	It remains to show that $C$ is not order closed. Suppose that it is order closed. Let $g$ be a nonnegative element in $\mathcal{X}\backslash\mathcal{X}_a$ and $\lambda\geq\int_{\Omega}g\mathrm{d}\mu$. Set
	\[
	f_n=\lambda\rchi_{\Omega}-g\rchi_{\{g\leq n\}}
	\]  
	for every $n\in\mathbb{N}$ and $f=\lambda\rchi_{\Omega}-g$. Clearly, $f_n\xrightarrow{o}f$. Since $\{f_n\} \subseteq L^\infty \subseteq\mathcal{X}_a$ and $\int_{\Omega}f_n\mathrm{d}\mu\geq \lambda-\int_{\Omega}g\mathrm{d}\mu\geq0$, we obtain that $\{f_n\} \subseteq C$ and hence, $f\in C$ by the order closedness of $C$. It follows that $f\in\mathcal{X}_a$ and hence, $g=\lambda\rchi_{\Omega}-f\in\mathcal{X}_a$, a contradiction. 
	
	$(2)\implies(1)$. If $\mathcal{X}=L^\infty$, then $(1)$ holds by \cite[Remark 4.4]{JST} and Theorem \ref{cvxinvariantequiv}. If $\mathcal{X}$ is order continuous, then $\sigma(\mathcal{X},\mathcal{X}_n^\sim)$ is just the weak topology on $\mathcal{X}$. It follows that any norm closed convex in $\mathcal{X}$ is $\sigma(\mathcal{X},\mathcal{X}_n^\sim)$-closed and hence, order closed.
\end{proof}

\section{Quasiconvex law-invariant functionals with the Fatou property}
We say that a functional $\rho:\mathcal{X}\to(-\infty,\infty]$ is proper if it is not identically $\infty$. It is convex (resp., quasiconvex) if $\rho(\lambda f+(1-\lambda) g)\leq \lambda\rho(f)+(1-\lambda)\rho(g)$ for any $f,g\in\mathcal{X}$ and $\lambda\in[0,1]$  (resp., the sublevel set $\left\{\rho\leq\lambda\right\}:=\left\{f\in\mathcal{X}:\rho(f)\leq\lambda\right\}$ is convex for every $\lambda\in\mathbb{R}$). It is law-invariant if $\rho(f)=\rho(g)$ whenever $f$ and $g$ have the same distribution. Note that $\rho$ is law-invariant if and only if each sublevel set $\{\rho\leq \lambda\}$ is law-invariant. If $\tau$ is a topology on $\mathcal{X}$, we say that $\rho$ is $\tau$-lower semicontinuous if the sublevel set $\left\{\rho\leq\lambda\right\}$ is $\tau$-closed for every $\lambda\in\mathbb{R}$.

Recall that a functional $\rho$ is said to have the Fatou property if $\rho(f)\leq\liminf_n\rho(f_n)$ whenever $f_n\xrightarrow{o}f$. Since the Fatou property is equivalent to order closedness of $\left\{\rho\leq\lambda\right\}$ for every $\lambda\in\mathbb{R}$, Theorem \ref{cvxinvariantequiv} gives us the following result which generalizes Theorem 1.1 in \cite{GLMX} and Proposition 11 in \cite{CGX}. This shows, in particular, that any proper convex law-invariant functional on $\mathcal{X}$ with the Fatou property admits the Fenchel-Moreau dual representation via $\mathcal{X}_n^\sim$.

\begin{theorem}\label{fatou-lsc-equiv}
	Let $\rho:\mathcal{X}\to(-\infty,\infty]$ be a proper quasiconvex law-invariant functional. Then the following are equivalent:
	\begin{enumerate}[$(1)$]
		\item $\rho$ has the Fatou property.
		\item $\rho$ is $\sigma(\mathcal{X},\mathcal{X}_n^\sim)$-lower semicontinuous.
		\item $\rho$ is $\sigma(\mathcal{X},L^\infty)$-lower semicontinuous.
	\end{enumerate}
	If $\mathcal{X}\neq L^1$, they are also equivalent to the following:
	\begin{enumerate}[$(1)$]\setcounter{enumi}{3}
		\item $\rho$ is $\sigma(\mathcal{X},\mathcal{X}_{uo}^\sim)$-lower semicontinuous.
	\end{enumerate}
\end{theorem}

Note that any functional with the Fatou property is automatically norm lower semicontinuous and the converse is also true when $\mathcal{X}$ is order continuous.  In \cite{JST}, Jouini et al. showed that for any proper convex law-invariant functional $\rho:L^\infty\to(-\infty,\infty]$, norm lower semicontinuity of $\rho$ implies the Fatou property. However, this  may fail in Orlicz spaces. In fact, Gao et al. in \cite{GLMX} proved that the same implication remains true for every proper quasiconvex law-invariant functional on an Orlicz space $L^\phi$ if and only if $\phi$ satisfies the $\Delta_2$ condition, or equivalently, $L^\phi$ is order continuous. We extend this result to general r.i. spaces.

\begin{theorem}
	Let $\mathcal{X}\neq L^\infty$ be not order continuous. There exists a proper norm lower semicontinuous convex law-invariant functional on $\mathcal{X}$ which fails the Fatou property.
\end{theorem}
\begin{proof} 
	Let $C$ be the set constructed in Theorem \ref{cvxinveqvnorm}. Define a functional $\rho:\mathcal{X}\to[-\infty,\infty]$ by 
	$\rho(f)=\inf\{m\in\mathbb{R}:f+m\rchi_{\Omega}\in C\}$ if $f+m\rchi_\Omega\in C$ for some $m\in\mathbb{R}$ and $\rho(f)=\infty$ elsewhere. Clearly, $\rho$ is proper and does not attain the value $-\infty$. Since $C$ is convex and law-invariant, $\rho$ is convex and law-invariant. Observe that $C\subseteq \{\rho\leq0\}$. If $\rho(f)<0$, then there exists $m<0$ such that $f+m\rchi_\Omega\in C$. Since $-m\rchi_{\Omega}\in C$, we obtain that $f=f+m\rchi_\Omega-m\rchi_\Omega\in C$. If $\rho(f)=0$, then $f$ lies in the norm closure of $C$. Since $C$ is norm closed, it follows that $f\in C$. Thus, we deduce that $C= \{\rho\leq0\}$. As a consequence, each sublevel set $\{\rho\leq\lambda\}=\{f\in\mathcal{X}:\rho(f+\lambda\rchi_{\Omega})\leq0\}$ is norm closed and hence, $\rho$ is norm lower semicontinuous. However, since $\{\rho\leq0\}=C$ is not order closed, $\rho$ fails the Fatou property. 
\end{proof}

We end this section by providing the relation between the Fatou property and another Fatou-type property introduced by Gao and Xantos \cite{GX} which is known as the strong Fatou property. A functional $\rho$ on $\mathcal{X}$ is said to have the strong Fatou property if $\rho(f)\leq\liminf_n\rho(f_n)$ whenever $f_n\xrightarrow{a.e}f$ and $\{f_n\}$ is norm bounded. Clearly, any functional with the strong Fatou property has the Fatou property. 	When $\mathcal{X}\neq L^1$ is an Orlicz space, the Fatou property and the strong Fatou property are equivalent for any proper quasiconvex law-invariant functional on $\mathcal{X}$ (\cite{GLMX}). When $\mathcal{X}=L^1$,  the  convex law-invariant functional $\rho:\mathcal{X}\to(-\infty,\infty]$,  defined by
\[
\rho(f)=\int_{\Omega}f\mathrm{d}\mu,\ f\in L^1,
\]
has the Fatou property but fails not the strong Fatou property (see \cite[Example 7]{CGX}). In case $\mathcal{X}\neq L^1$ is an order continuous r.i. space, Chen et al. \cite{CGX} showed that the same equivalence remains true. The following theorem shows that in fact the equivalence holds for any r.i. space $\mathcal{X}\neq L^1$.

\begin{theorem}
	Let $\mathcal{X}\neq L^1$. For any proper quasiconvex law-invariant functional $\rho:\mathcal{X}\to(-\infty,\infty]$, the Fatou property and the strong Fatou property of $\rho$ are equivalent.
\end{theorem}
\begin{proof}
	It is enough to show that any proper quasiconvex law-invariant functional  $\rho:\mathcal{X}\to(-\infty,\infty]$ with the Fatou property has the strong Fatou property. Let $\{f_n\}$ be a norm bounded sequence in $\mathcal{X}$ that converges a.e. to $f\in \mathcal{X}$. Then  $f_{n}\xrightarrow{\sigma(\mathcal{X},\mathcal{X}_{uo}^\sim)}f$. Since $\rho$ is $\sigma(\mathcal{X},\mathcal{X}_{uo}^\sim)$-lower semicontinuous (Theorem \ref{fatou-lsc-equiv}), it follows that $\rho(f)\leq\liminf_n\rho(f_{n})$. Thus, $\rho$ has the strong Fatou property.
\end{proof}

\end{document}